\documentclass{article}
\usepackage[final,nonatbib]{neurips_2018}

\usepackage{subfig,float,wrapfig,caption}
\usepackage{graphicx}
\usepackage{graphics}
\usepackage{amssymb}
\usepackage{amsmath}
\usepackage{color}
\usepackage{soul}
\usepackage{xspace}
\usepackage{algpseudocode}
\usepackage{bbm}
\usepackage{comment}
\usepackage{algorithmicx}
\usepackage{psfrag,url}
\usepackage{cite}

\newcommand{\vect}[1]{\boldsymbol{\mathbf{#1}}}

\newcommand{\real}{{\mathbb{R}}}

\usepackage{cite,comment}
\usepackage{wrapfig}
\usepackage{amsthm}
\usepackage{mathrsfs}
\usepackage{graphicx}
\usepackage{graphics}
\usepackage{amssymb}
\usepackage{amsmath,mathtools}
\usepackage{color}
\usepackage{xspace}
\usepackage{algpseudocode}
\usepackage{bbm}
\usepackage{comment}
\usepackage{algorithmicx}

\usepackage{rotating}
\usepackage{chngcntr}
\usepackage{apptools}
\AtAppendix{\counterwithin{thm}{section}}
\usepackage{graphicx}
\usepackage{graphics}
\usepackage{amssymb}
\usepackage{amsmath}
\usepackage{mathtools}
\usepackage{color}
\usepackage{xspace}
\usepackage{algpseudocode}
\usepackage{bbm}
\usepackage{comment}
\usepackage{algorithmicx}
\usepackage{subfig}
\usepackage{psfrag}

\usepackage{rotating}
\usepackage{chngcntr}
\usepackage{apptools}
\usepackage{listings}
\AtAppendix{\counterwithin{thm}{section}}
\usepackage{graphicx}
\usepackage{graphics}
\usepackage{amssymb}
\usepackage{amsmath}
\usepackage{color}
\usepackage{xspace}
\usepackage{algpseudocode}
\usepackage{bbm}
\usepackage{comment}
\usepackage{algorithmicx}
\usepackage{subfig}
\usepackage{psfrag}

\usepackage{enumitem}

\usepackage{graphics}
\usepackage{amssymb}
\usepackage{amsmath}
\usepackage{color}
\usepackage{xspace}
\usepackage{algpseudocode}
\usepackage{bbm}
\usepackage{comment}
\usepackage{algorithmicx}
\usepackage{subfig}
\usepackage{psfrag}
\usepackage{soul}

\usepackage{cancel}
\newtheorem{theorem}{Theorem}[section]
\newtheorem{remark}{Remark}[section]

\newtheorem{lemma}[theorem]{Lemma}

\newtheorem{definition}{Definition}

\newcommand{\oprocendsymbol}{\hbox{$\bullet$}}
\newcommand{\oprocend}{\relax\ifmmode\else\unskip\hfill\fi\oprocendsymbol}

\newcommand{\VV}{\mathcal{V}}
\newcommand{\EE}{\mathcal{E}}
\newcommand{\GG}{\mathcal{G}}

\newcommand{\lL}{\vect{\mathsf{L}}}

\newcommand{\re}[1]{\operatorname{Re}(#1)}

\newcommand{\integernonnegative}{{\mathbb{Z}}_{\ge0}}

\newcommand{\review}[1]{{\color{black}#1}}

\newcommand{\vectsf}[1]{\boldsymbol{\mathbf{\mathsf{#1}}}}

\newcommand{\Diag}[1]{\operatorname{Diag}(#1)}

 \newcommand{\boxend}{\hfill \ensuremath{\Box}}
 
\usepackage{makecell}
\usepackage{graphicx}

\usepackage[normalem]{ulem}

\usepackage[utf8]{inputenc} 
\usepackage[T1]{fontenc}    
\usepackage{hyperref}       
\usepackage{url}            
\usepackage{booktabs}       
\usepackage{amsfonts}       
\usepackage{nicefrac}       
\usepackage{microtype}      

\title{Deterministic Privacy Preservation in \\Static Average Consensus Problem}

\author{%
  Amir-Salar~Esteki\\
  Dept. of Mechanical and Aerospace Eng.\\
  University of California, Irvine\\
  Irvine, CA 92617 \\
  \texttt{aesteki@uci.edu} \\
  \And
  Solmaz S.~Kia\\
  Dept. of Mechanical and Aerospace Eng.\\
  University of California, Irvine\\
  Irvine, CA 92617 \\
  \texttt{solmaz@uci.edu} \\
}

\usepackage{setspace}
\begin{document}

\maketitle
\begin{abstract}
In this paper we consider the problem of privacy preservation in the static average consensus problem. This problem normally is solved by proposing privacy preservation augmentations for the popular first order Laplacian-based algorithm. These mechanisms however come with computational overhead, may need coordination among the agents to choose their parameters and also alter the transient response of the algorithm. In this paper we show that an alternative iterative algorithm that is proposed in the literature in the context of dynamic average consensus problem has intrinsic privacy preservation and can be used as a privacy preserving algorithm that yields the same performance behavior as the well-known Laplacian consensus algorithm but without the overheads that come with the existing privacy preservation methods.
\end{abstract}

\section{Introduction}
This paper considers the problem of privacy preservation in the in-network average consensus (PrivCon) problem, where the objective is to enable a group of $\mathcal{V}\!=\!\{1,\cdots,N\}$ communicating agents interacting over a strongly connected and weight-balanced digraph $\mathcal{G}(\VV,\EE,\vectsf{A})$\footnote{For graph theoretic definitions used hereafter see Section~\ref{sec::prob_setting}.}, see Fig.~\ref{fig::digraph}, to calculate  $\mathsf{r}^\text{avg}\!=\!\frac{1}{N}\sum_{j=1}^N \mathsf{r}^j$ using local interactions but without disclosing their local \emph{reference value} $\mathsf{r}^i$, $i\!\in\!\mathcal{V}$, to each other. The privacy preservation for average consensus is normally formalized as concealing the reference value of each agent from a malicious agent that is defined as~follows.
\begin{definition}
    [Malicious agent]\label{asm::knowledgesetadversary}{\rm
    A malicious agent in the network is an agent that wants to obtain the reference value of the other agents without perturbing/interrupting the execution of the average consensus algorithm. The knowledge set of this malicious agent consists of (a) the network topology  $\mathcal{G}(\VV,\EE,\vectsf{A})$, (b) its own local states and reference input, (c) the received signals from its out-neighbors, and (d) the agreement state $x^i(k)$ of each agent $i\in\VV$ converges asymptotically to $\mathsf{r}^{\text{avg}}$}.\boxend
\end{definition}
The solutions to the PrivCon problem in the literature mainly center around designing privacy preservation augmentation mechanisms for the popular average consensus algorithm
\begin{align}\label{eq::Alg1}
    &x^i(k+1)=x^i(k)-\Delta\sum\nolimits_{j=1}^{N}\vectsf{a}_{ij}(x^{i}(k)-x^{j}(k)),\\
    &x^i(0)=\mathsf{r}^i, \quad i\in\VV\nonumber,
\end{align}
which with proper choice of stepsize $\Delta$ guarantees $x^i\!\to\!\mathsf{r}^\text{avg}$, $i\in\VV$, as $k\!\to\!\infty$~\cite{olfati2004consensus}. 
In a weighted digraph, $\vectsf{a}_{ij}>0$ if there is an edge from agent $i$ to agent $j$, i.e., agent $i$ can obtain information from agent $j$, otherwise $\vectsf{a}_{ij}=0$. Thus, algorithm~\eqref{eq::Alg1} requires each agent $i\in\VV$ to share its reference value $\mathsf{r}^i$ with its in-neighbors (agents that receive messages from $i$) in the first step of the algorithm, resulting in a trivial breach of privacy. Standard observability analysis \cite{kalman1960general} shows also that when the adjacency matrix $\vectsf{A}=[\vectsf{a
}_{ij}]$ of the network topology is known to an agent $i\in\VV$, agent $i$ can  obtain the initial condition of algorithm~\eqref{eq::Alg1}, and consequently, the reference value of all or some of the other agents of the network; see~\cite{pequito2014design} for details.

A popular approach explored to induce privacy preservation to algorithm~\eqref{eq::Alg1} is the encryption  where either a trusted third-party generates the public-key~\cite{hadjicostis2018privary} or agents share their keys in which the network is restricted to a point-to-point or tagged undirected communication framework~\cite{MR-HG-YW:19, yin2019accurate}. Also,~\cite{MR-HG-YW:19} requires extra communication channels for key generation purposes. Moreover, in~\cite{MR-HG-YW:19, yin2019accurate}, the network should be connected at the first step of communication and therefore, it does not have robustness to possible switching in the topology.
Differential privacy~\cite{nozari2017differentially, huang2012differentially,nozari2015differentially} and additive obfuscation noise methods~\cite{manitara2013privacy, mo2016privacy,he2018privacy} are other popular techniques to conceal the \emph{exact} reference value of the agents, but the malicious agent can obtain an estimate on the reference value using a stochastic estimator. In addition, final convergence is perturbed in~\cite{nozari2017differentially, huang2012differentially,nozari2015differentially, manitara2013privacy} and convergence rate is altered in~\cite{mo2016privacy}. Moreover, these classes of privacy preserving measures cause noisy transient response for the algorithm, which results in excessive energy expenditure if agents use local controls to track the consensus trajectory as a reference input. Lastly, the guarantees established in all the methods discussed above are only for connected undirected graphs. Interested reader can also find privacy preservation methods for the continuous-time implementation of algorithm~\eqref{eq::Alg1} in~\cite{altafini2020system} and~\cite{NR-SSK:18}.

\begin{figure}[t]
\centering
{\includegraphics[width=0.75\linewidth]{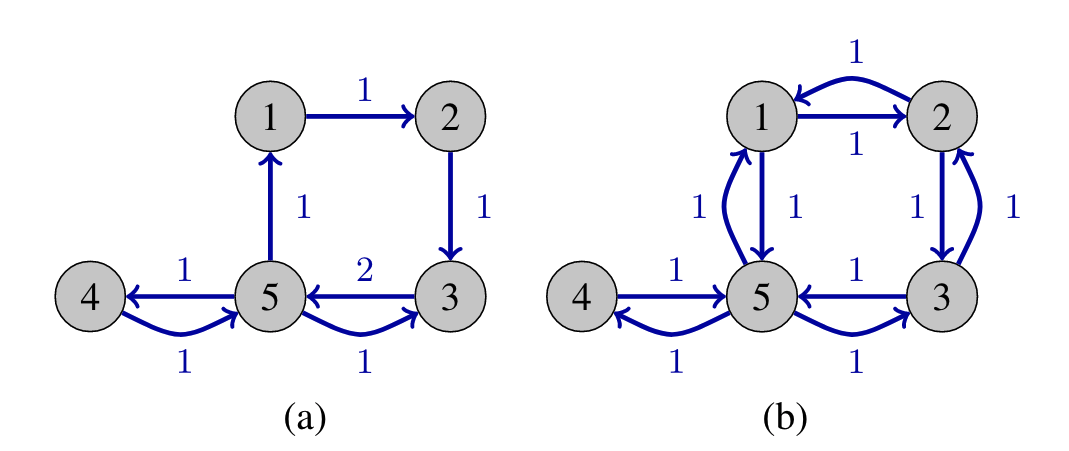}}
\caption{{\small The network in plot (a) is an example of a strongly~connected and weight-balanced digraph, while the network in plot (b) is an example of a connected undirected graph. An edge $i\xrightarrow{\mathit{\vectsf{a}_{ij}}}j$ from an agent $i$ to agent $j$ means that agent $i$ can obtain~information from agent $j$; $\vectsf{a}_{ij}\!>\!0$ is the corresponding adjacency matrix~element.
}
}
\label{fig::digraph}
\end{figure}

In this paper, we focus on a solution for PrivCon problem that does not perturb the final convergence value.  We address the PrivCon problem by proposing to use an alternative algorithm (see algorithm~\eqref{eq::Alg2}) that yields similar transient and convergence performance to that of algorithm~\eqref{eq::Alg1}. We analyze the privacy preservation of this algorithm carefully and show that this algorithm intrinsically yields exactly the same privacy preservation guarantees as in \cite{manitara2013privacy,mo2016privacy,he2018privacy} in terms of which agents' privacy can be preserved. More precisely, we show that similar to the results in~\cite{manitara2013privacy,mo2016privacy,he2018privacy} the privacy of any agent that has at least one out-neighbor that is not the out-neighbor of the malicious agent is preserved. In comparison to the privacy preservation by use of the additive obfuscation noise methods~\cite{manitara2013privacy, mo2016privacy,he2018privacy}, however, algorithm~\eqref{eq::Alg2} offers a deterministic and stronger sense of privacy preservation, i.e., it meets the privacy preservation definition given~below. 
\begin{definition}
    [Privacy preservation]\label{asm::privacy}{\rm
    Privacy of an agent is preserved from a malicious agent if the malicious agent~cannot obtain any estimate on the reference value of the agent.~\boxend
    }
\end{definition}
To show that the privacy of an agent $i$ is preserved in the sense of Definition~\ref{asm::privacy}, we show that there exists other values of $\mathsf{r}^i$, \emph{arbitrarily} different from the true one, that give the same trajectory for the information the malicious agent receives from its out-neighbors. Therefore, the malicious agent cannot distinguish between the true reference value and the arbitrary alternatives. This, in essence, is a stronger notion of privacy preservation than that of the $\epsilon$-differential privacy~\cite{dwork2014algorithmic} or that of~\cite{mo2016privacy} where even though the malicious agent cannot obtain the exact reference value of the agents but it obtains an estimate on the reference value.


In our demonstration study in Section~\ref{sec::num}, we compare in particular our suggested PrivCon problem solution to the additive obfuscation noise method of~\cite{mo2016privacy}.~\cite{mo2016privacy} is a prominent solution in the class of additive obfuscation noise methods that guarantees exact convergence while not allowing a malicious agent to obtain the exact reference value of the agents under the connectivity condition stated earlier. However, in the framework of~\cite{mo2016privacy} there are disclosure of information about the reference value in two levels. First, as shown in~\cite{mo2016privacy}, the malicious agent can employ a stochastic observer to obtain an estimate on the reference value of other agents. The normalized variance of this estimator can be very small compared to the size of the reference value of the agents. Moreover, in~\cite{mo2016privacy} agents need to coordinate to choose a common parameter $\phi$ in their noise generator and also use a common variance. The knowledge about the noise parameters also is another level of disclosure of an estimate on the reference value of the agents, as the transmit message of each agent at the initial step of the algorithm is the reference value plus  a value generated by their local noise whose variance is the same for all the agents. Moreover, in choosing the variance value for the noise, agents need to consider the size of their local reference values to yield meaningful obfuscation. We note that using a noise with large variance results in a more perturbed transient response and a slower convergence. The advantage of our proposed solution is to meet the privacy preservation as given in Definition~\ref{asm::knowledgesetadversary} while rendering a similar transient and convergence behavior to that of algorithm~\eqref{eq::Alg1}, having no need for coordination to choose the parameters of the algorithm, and no extra computations. Our notation is standard, though certain pieces of notation are defined as need arises.
\section{Problem Setting}\label{sec::prob_setting}
Our graph theoretic definitions follow~\cite{bullo2009distributed}. A weighted digraph is a triplet $\mathcal{G}(\VV,\EE,\vectsf{A})$
 where $\VV$ is the node set, $\EE\subset\VV\times \VV$ is the edge set and $\vectsf{A})=[\vectsf{a}_{ij}]$ is the
weighted adjacency matrix defined such that $\vectsf{a}_{ij}>0$ if $(i,j)\in\mathcal{E}$, otherwise $\vectsf{a}_{ij}=0$. Given an edge $(i,j)$, $i$ is called 
 an in-neighbor of 
$j$, and 
$j$ is called an out-neighbor of $i$. 
 The (weighted) in-degree and out-degree of 
 each node $i\in\VV$ are respectively are $\mathsf{d}_{\text{out}}^i=\sum_{j=1}^N \vectsf{a}_{ij}$ and $\mathsf{d}_{\text{in}}^i=\sum_{j=1}^N \vectsf{a}_{ji}$. A weighted digraph is undirected if $\vectsf{a}_{ij}=\vectsf{a}_{ji}$
 for all $i,j\in\VV$. A digraph is weight-balanced if at every node $i\in\VV$, $\mathsf{d}_{\text{out}}^i=\mathsf{d}_{\text{in}}^i$. A digraph is strongly connected if there is a directed path from every node to every other node. The Laplacian matrix of a digraph $\GG$ is  $\lL\!=\!\Diag{\mathsf{d}_{\text{out}}^1,\cdots, \mathsf{d}_{\text{out}}^N}\!-\!\vectsf{A}$.

Let the interaction topology of the agents be a strongly connected and weight-balanced digraph $\GG\!=\!(\VV,\mathcal{E},\vectsf{A})$. We consider the average consensus algorithm 
\begin{subequations}\label{eq::Alg2}
\begin{align}
    v^i(k+1)&=v^i(k)+\Delta\sum\nolimits_{j=1}^{N}\vectsf{a}_{ij}(x^{i}(k)-x^{j}(k)),\label{eq::Alg2a}\\
    x^i(k+1)&=x^i(k)+\Delta\big(-(x^i(k)-\mathsf{r}^i)\nonumber\\
    &-\sum\nolimits_{j=1}^{N}\vectsf{a}_{ij}(x^{i}(k)-x^{j}(k))-v^i(k)\big),\label{eq::Alg2b}\\
    &\!\!\!\!\!\!\!\!\!\!\!\!\!\!\!\!\!\!\!\!\!\!x^i(0), v^i(0)\in\real, ~ i\in\VV, ~\text{s.t.}~\sum\nolimits_{i=1}v^i(0)=0,
\end{align}
\end{subequations}
proposed in~\cite{SSK-JC-SM:15-ijrnc}, as a solution for privacy preserving average consensus algorithm over $\GG$. As discussed in~\cite{kia2019tutorial}, 
~\eqref{eq::Alg2} yields a similar transient response to that of algorithm~\eqref{eq::Alg1}.
In our analysis, the privacy preservation in the sense of Definition~\ref{asm::privacy} is against a malicious agent defined as in Definition~\ref{asm::knowledgesetadversary}.  \review{Our study is different than the privacy preservation study in~\cite{SSK-JC-SM:15-ijrnc} 
where the continuous-time representation of algorithm~\eqref{eq::Alg2} is considered and the reference values are assumed to be dynamic time-varying signals. The guarantees provided crucially depend on the time varyingness of the reference values.}

Let $\vect{x}=(x^1,\cdots,x^N)$, $\vect{v}=(v^1,\cdots,v^N)$,    $\vect{\Pi}=\textbf{I}_{N}-\frac{1}{N}\vectsf{1}_{N}\vectsf{1}_{N}^\top$, and  $\bar{\vectsf{r}}=\mathsf{r}^{\text{avg}}\vectsf{1}_N$. Moreover, let $\vect{\mathfrak{r}}=\frac{1}{\sqrt{N}}\vect{1}_N$, and $ \vectsf{\mathfrak{R}}$ such that $$\begin{bmatrix}\vect{\mathfrak{r}}&~\vectsf{\mathfrak{R}}\end{bmatrix}\begin{bmatrix}\vect{\mathfrak{r}}^\top\\\vect{\mathfrak{R}}^\top\end{bmatrix}=\begin{bmatrix}\vect{\mathfrak{r}}^\top\\\vect{\mathfrak{R}}^\top\end{bmatrix}\begin{bmatrix}\vect{\mathfrak{r}}&~\vectsf{\mathfrak{R}}\end{bmatrix}=\vect{I}_N.$$ Then, using  the change of variable $$
\begin{bmatrix}
    \mathsf{q}_1\\
    \vectsf{q}_{2:N}\\
    \mathsf{p}_1\\
    \vectsf{p}_{2:N}
\end{bmatrix}
=
\begin{bmatrix}
    \begin{bmatrix}
    \vectsf{\mathfrak{r}}^\top \\
    \vectsf{\mathfrak{R}}^\top
    \end{bmatrix} &
    \begin{bmatrix}
    \vectsf{0} \\
    \vectsf{\mathfrak{R}}^\top
    \end{bmatrix}\\
    \vectsf{0} &
    \begin{bmatrix}
    \vect{\mathfrak{r}}^\top \\
    \vectsf{\mathfrak{R}}^\top
    \end{bmatrix}
\end{bmatrix}
\begin{bmatrix}
    \vectsf{v}-\vect{\Pi} \vectsf{r}\\
    \vectsf{x}-\bar{\vectsf{r}}
\end{bmatrix},$$
algorithm~\eqref{eq::Alg1} reads~as $$\begin{bmatrix}
    \mathsf{p}_1(k+1)\\
    \vectsf{p}_{2:N}(k+1)
\end{bmatrix}=\begin{bmatrix}
    1&\vectsf{0}\\
    0&~\vect{I}-\Delta\lL^{+}
\end{bmatrix}\begin{bmatrix}
    \mathsf{p}_1(k)\\
    \vectsf{p}_{2:N}(k)
\end{bmatrix},$$ and algorithm~\eqref{eq::Alg2}~as
\begin{align*}
   \begin{bmatrix}
    \mathsf{q}_1(k+1)\\
     \vectsf{q}_{2:N}(k+1)\\
     \mathsf{p}_1(k+1)\\
     \vectsf{p}_{2:N}(k+1)
   \end{bmatrix}\!\!=\!\!\begin{bmatrix}1&\vect{0}&0&\vect{0}\\
   0&(1-\Delta)\vect{I}&0&\vect{0}\\
   -\Delta&\vect{0}&(1-\Delta)&\vect{0}\\
   0&-\Delta\vect{I}&0&\vect{I}-\Delta\lL^{+}\end{bmatrix}\!\!\!\begin{bmatrix}
    \mathsf{q}_1(k)\\
     \vectsf{q}_{2:N}(k)\\
     \mathsf{p}_1(k)\\
     \vectsf{p}_{2:N}(k)
   \end{bmatrix},
\end{align*}
where $\lL^{+}\!=\!\vectsf{\mathfrak{R}}^\top\vect{L}\vectsf{\mathfrak{R}}$. These equivalent representations show the connection between the internal dynamics of algorithms~\eqref{eq::Alg1} and~\eqref{eq::Alg2}. For a strongly connected and weight-balanced digraph, $\lL$ has a simple zero eigenvalue and the rest of the eigenvalues $\{\lambda_i\}_{i=2}^N$ have non-zero positive real parts. 
Then, the eigenvalues of $\lL^{+}$ are $\{\lambda_i\}_{i=2}^N$. 
By invoking \cite[Lemma~S1]{kia2019tutorial}, we note that the exponential stability of~\eqref{eq::Alg1} and~\eqref{eq::Alg2} is guaranteed, respectively, for any $\Delta\!\in\!(0,\bar{\Delta})$ and $\Delta\!\in\!(0,\min\{2,\bar{\Delta}\})$, where $\bar{\Delta}\!=\!\min\big\{2\frac{\re{\lambda_i}}{|\lambda_i|^2}\big\}_{i=2}^N$.
Given a similar performance, an immediate appeal of algorithm~\eqref{eq::Alg2} over algorithm~\eqref{eq::Alg1} is that the reference value of agents is not trivially \review{transmit} to the in-neighbors. The question that we address in the subsequent section is whether the malicious agent can use its knowledge set to compute the reference values of other~agents when agents implement algorithm~\eqref{eq::Alg2}.
\section{Privacy Preservation Analysis}
In this section we show that algorithm~\eqref{eq::Alg2} intrinsically yields the same privacy preservation guarantees as in~\cite{manitara2013privacy,mo2016privacy,he2018privacy} in terms of which agents' privacy can be preserved. However, the guarantees of algorithm~\eqref{eq::Alg2} are valid also for strongly connected and weight-balanced digraphs. Moreover, unlike the privacy preservation of~\cite{manitara2013privacy, mo2016privacy,he2018privacy} which is obtained by using additive noises and comes with disclosure of a stochastic estimate on the private value of the agents, algorithm~\eqref{eq::Alg2} offers a deterministic and stronger sense of privacy preservation, i.e., it meets the privacy preservation objective defined in Definition~\ref{asm::privacy}. 
In what follows without loss of generality we assume that agent $1$ is the malicious agent.  
The initialization condition $\sum_{i=1}^Nv^i(0)=0$ is trivially satisfied when every agent $i\in\VV$  uses $v^i(0)\!=\!0$. Other choices need coordination among agents with no strong guarantee that the choices are private. Thus
, we assume that $v^i(0)=0$, $i\in\VV$ and is known to the malicious agent. Moreover, $\mathcal{N}_{\text{in}}^i$ and $\mathcal{N}_{\text{out}}^i$ are, respectively, the set of in-neighbors and out-neighbors of agent $i$. We also define  $\mathcal{N}_{\text{in}}^{i+}=\mathcal{N}_{\text{in}}^i\cup\{i\}$.

\begin{lemma}[A sufficient condition for privacy preservation when algorithm~\eqref{eq::Alg2} is implemented]\label{lem::PrivacyAlg2}
    Let the interaction topology $\GG$ of the agents implementing algorithm~\eqref{eq::Alg2} with $\Delta\in(0,\min\{2,\bar{\Delta}\})$, initialized at $x^i(0)\in\real$ and $v^i(0)=0$, $i\in\VV$, be strongly connected and weight-balanced digraph. Let the knowledge set of malicious agent $1$ be as in Definition~\ref{asm::knowledgesetadversary}.  The privacy of any agent $i\in\VV\backslash\{1\}$ is preserved from agent $1$ if agent $i$ has at least one out-neighbor that is not an out-neighbor of agent~$1$.
\end{lemma}
\begin{proof}
To prove the statement, we consider the worst case where agent $i$ whose privacy is being evaluated is an out-neighbor of agent $1$, i.e., agent $1$ receives the \review{transmitted} message of agent $i$. Without loss of generality let this agent be labeled agent $2$. Let agent $3$ be the out-neighbor of agent $2$ that is not an out-neighbor of agent $1$, i.e.,  $2\in\mathcal{N}_{\text{in}}^3$ but $1\not\in\mathcal{N}_{\text{in}}^3$. To show privacy preservation for agent $2$, we show that there are infinitely many execution\review{s} of algorithm~\eqref{eq::Alg2} with different values of $\mathsf{r}^2\in\real$ and $\mathsf{r}^3\in\real$  that yield exactly the same received signal by agent $1$.  To this end, consider two executions of algorithm~\eqref{eq::Alg2}. 
Let $e_x^i=x_1^i-x_2^i$ and $e_v^i=v_1^i-v_2^i$, $i\in\VV$. Moreover, let the reference value of the agents in the first execution be $\{\mathsf{r}_1^i\}_{i=1}^N$ and in the second execution be $\{\mathsf{r}_2^i\}_{i=1}^N$ such that $e^i_{\mathsf{r}}=\mathsf{r}_1^i-\mathsf{r}_2^i\neq0$ if $i\in\mathcal{N}_{\text{in}}^{3+}$, otherwise $e^i_{\mathsf{r}}=0$. Moreover, $\sum\nolimits_{i=1}^{N}e^i_{\mathsf{r}}=0$, i.e., $\mathsf{r}_1^{\text{avg}}=\mathsf{r}_2^{\text{avg}}$. Lastly, $e_x^i(0)=0$, for $i\in\VV\backslash \{3\}$. 
Next, we show that for any $e_x^3(0)\in\real$ there exists $e_{\mathsf{r}}^i\neq0$, $i\in\mathcal{N}_{\text{in}}^{3+}$, such that $e_x^j(k)\equiv0$, $j\in\VV\backslash\{3\}$ for $k\in\integernonnegative$. Therefore, the signals received by agent $1$ for these two distinct executions is exactly the same and agent $1$ cannot distinguish between them. 

The error dynamics at each agent $i\in\VV$ reads as
\begin{subequations}\label{eq::Alg2_eror}
\begin{align}
    e_v^i(k+1)&=e_v^i(k)+\Delta\sum\nolimits_{j=1}^{N}\vectsf{a}_{ij}(e_x^{i}(k)-e_x^{j}(k)),\label{eq::Alg2a_e}\\
    e_x^i(k+1)&=e_x^i(k)+\Delta\big(-(e_x^i(k)-e_\mathsf{r}^i)\nonumber\\
    &-\sum\nolimits_{j=1}^{N}\vectsf{a}_{ij}(e_x^{i}(k)-e_x^{j}(k))-e_v^i(k)\big),\label{eq::Alg2b_e}\\
    &\!\!\!\!\!\!\!\!\!\!\!\!\!\!\!\!\!\!\!\!\!\!  e_v^i(0)=0,~~ \begin{cases}e_x^i(0)=0,&\text{if}~i\neq3,\\
    e_x^i(0)\neq 0&\text{if}~i=3
    \end{cases}.
\end{align}
\end{subequations}
If $e_x^i(k)\equiv0$ for any $i\in\VV\backslash\{3\}$, then 
 $e_x^i(k)\equiv0$, and $e_v^i(k)\equiv0$ for $k\in\integernonnegative$ satisfy the error dynamics~\eqref{eq::Alg2_eror} for any $i\in\VV\backslash\mathcal{N}_{\text{in}}^{3+}$. On the other hand, for $i\in\mathcal{N}_{\text{in}}^{3}$ 
 \begin{subequations}\label{eq::PrivacyErrorRelationsAlg2}
    \begin{align}
        e^{i}_{v}(k+1)&=e^{i}_{v}(k)-\Delta \vectsf{a}_{i3}\, e^{3}_{x}(k),\label{eq::PrivacyErrorRelationsAlg2c}\\
        0&=e^{i}_{\mathsf{r}}+\vectsf{a}_{i3}\,e^{3}_{x}(k)-e^{i}_{v}(k),\label{eq::PrivacyErrorRelationsAlg2d}
    \end{align}
\end{subequations}
 and since $\sum_{i=1}^N\vectsf{a}_{3i}=\mathsf{d}_\text{out}^3$,~\eqref{eq::Alg2_eror} for agent $3$ reads as
\begin{subequations}\label{eq::PrivacyErrorRelationsAlg2-ag3}
    \begin{align}
        e^{3}_{v}(k+1)&=e^{3}_{v}(k)+\Delta\mathsf{d}_{\text{out}}^{3}e^{3}_{x}(k),\label{eq::PrivacyErrorRelationsAlg2a}\\
        e^{3}_{x}(k+1)&=e^{3}_{x}(k)\!+\!\Delta(e^{3}_{\mathsf{r}}\!-\!(1\!+\!\mathsf{d}_{\text{out}}^{3})e^{3}_{x}(k)-e^{3}_{v}(k)).\label{eq::PrivacyErrorRelationsAlg2b}
    \end{align}
\end{subequations}
Note that with adding \eqref{eq::PrivacyErrorRelationsAlg2c} and \eqref{eq::PrivacyErrorRelationsAlg2a} we have $\sum_{i\in\mathcal{N}^{3+}_{\text{in}}}e_v^i(k)=0$. Considering $\sum_{i=1}^N\vectsf{a}_{i3}=\mathsf{d}_\text{in}^3=\mathsf{d}_\text{out}^3$ and $\sum_{i\in\mathcal{N}^{3+}_{\text{in}}}e_{\mathsf{r}}^i=0$ also, it follows from adding up~\eqref{eq::PrivacyErrorRelationsAlg2b} and~\eqref{eq::PrivacyErrorRelationsAlg2d} for $i\in\mathcal{N}^{3+}_{\text{in}}$ that 
\begin{align}
e^{3}_{x}(k+1)&=(1-\Delta) e^{3}_{x}(k).\label{eq::e_3_evalve}
\end{align}
Then, since $e_v^3(0)\!=\!0$, it follows from~\eqref{eq::PrivacyErrorRelationsAlg2b},~\eqref{eq::e_3_evalve}, and~\eqref{eq::PrivacyErrorRelationsAlg2d} that 
\begin{align}\label{eq::privacyconstraints}
    e^{3}_{\mathsf{r}}=\mathsf{d}_{\text{out}}^{3}e^{3}_{x}(0),\quad  e^{i}_{\mathsf{r}}=-\vectsf{a}_{i3}e^{3}_{x}(0),\quad i\in\mathcal{N}_\text{in}^{3}.
\end{align}
Since $\sum_{i=1}^N\vectsf{a}_{i3}=\mathsf{d}_\text{in}^3=\mathsf{d}_\text{out}^3$, it follows from~\eqref{eq::privacyconstraints} that indeed $\sum\nolimits_{i\in\mathcal{N}_{\text{in}}^{3+}}e^i_{\mathsf{r}}=0$. Therefore, we can conclude that $\mathsf{r}_1^{\text{avg}}=\mathsf{r}_2^{\text{avg}}$. It is then proven that with an unbounded error in the initial state of $e_x^3$, $e^2_{\mathsf{r}}$ varies unboundedly as well yet the malicious agent receives the same signals from its out-neighbors, i.e., $e_x^i(k)\equiv0$, $k\in\mathbb{Z}_{\geq0}$, for $i\in\mathcal{N}_{\text{out}}^1$.
Note here that since $|1-\Delta|<1$, it follows from \eqref{eq::e_3_evalve} that $e_x^3(k)$ also converges to zero as $k\to\infty$ showing that all the agents converge to the same final average point in both execution $1$ and $2$.
\end{proof}
Our next result shows that 
when the malicious agent has access to the messages \review{transmitted} to and from an agent $i$, similar to the methods proposed in~\cite{manitara2013privacy,mo2016privacy,he2018privacy}, the privacy of agent $i$ is not preserved. To establish this result we show that the malicious agent can implement an observer to obtain $\mathsf{r}^i$.
\begin{lemma}[A sufficient condition for loss of privacy preservation when algorithm~\eqref{eq::Alg2} is implemented]\label{lem::PrivacyAlg2_observ}
    Let the interaction topology $\GG$ of the agents implementing algorithm~\eqref{eq::Alg2}, initialized at $x^i(0)\in\real$ and $v^i(0)=0$, $i\in\VV$, be \review{a} strongly connected and weight-balanced digraph. Let the knowledge set of malicious agent $1$ be as in Definition~\ref{asm::knowledgesetadversary}. If agent $1$ is the in-neighbor of agent $o\in\VV\backslash\{1\}$ and all its out-neighbors, then at any $k\in\mathbb{Z}_{\geq1}$ agent $1$ can obtain $\mathsf{r}^o$ using $x^o(k)$, $x^o(k-1)$ and $\{x^j(k-1)\}_{j\in\mathcal{N}_{\text{out}}^o}$.
\end{lemma}
\begin{proof}
Proof of Lemma~\ref{lem::PrivacyAlg2_observ} follows from a simple algebraic manipulation over~\eqref{eq::Alg2} to obtain $\mathsf{r}^o$ in terms of $x^o(k)$, $x^o(k-1)$ and $\{x^j(k-1)\}_{j\in\mathcal{N}_{\text{out}}^o}$.
\end{proof}

 Lemma~\ref{lem::PrivacyAlg2} and~\ref{lem::PrivacyAlg2_observ} gives us the following main statement on the privacy preservation guarantees of \review{Algorithm}~\eqref{eq::Alg2}.
\begin{theorem}[Privacy Preservation guarantee of algorithm~\eqref{eq::Alg2}]\label{thm::PrivacyAlg2}
 Let the interaction topology $\GG$ of the agents implementing algorithm~\eqref{eq::Alg2}, initialized at $x^i(0)\in\real$ and $v^i(0)=0$, $i\in\VV$, be strongly connected and weight-balanced digraph. Let the knowledge set of malicious agent $1$ be as in Definition~\ref{asm::knowledgesetadversary}.  Privacy of agent $i\in\VV\backslash\{1\}$ is preserved from agent $1$ if and only if agent $i$ has an out-neighbor that is not an out-neighbor of agent 1.
\end{theorem}
\begin{remark}[Privacy preservation for an initialization free average consensus algorithm for connected undirected graphs
]\label{rmk::PrivacyAlg3}
{\rm
We can show through similar arguments that the privacy preservation guarantees of algorithm~\eqref{eq::Alg2} hold for the average consensus algorithm~\eqref{eq::Alg3} as well,
\begin{subequations}\label{eq::Alg3}
\begin{align}
    v^i(k+1)=\,&v^i(k)-\Delta\sum\nolimits_{j=1}^{N}\vectsf{a}_{ij}(x^{i}(k)-x^{j}(k)),\label{eq::Alg3a}\\
    x^i(k+1)=\,&x^i(k)+\Delta\big(-(x^i(k)-\mathsf{r}^i)\nonumber\\
    &-\sum\nolimits_{j=1}^{N}\vectsf{a}_{ij}(x^{i}(k)-x^{j}(k))+\sum\nolimits_{j=1}^{N}\vectsf{a}_{ij}(v^{i}(k)-v^{j}(k))\big),\label{eq::Alg3b}\\
    &~~~x^i(0), v^i(0)\in\real, ~ i\in\VV.
\end{align}
\end{subequations}
Algorithm~\eqref{eq::Alg3} does not require special initialization of $v^i(0)=0$, $i\in\VV$, therefore it is robust to agent drop-off. However, this algorithm works over connected undirected graphs, and also requires an extra message exchange between the neighboring agents. 
}\boxend
\end{remark}
So far, we have shown the intrinsic privacy property of Algorithm~\eqref{eq::Alg2}. Now the natural question is whether we can loosen the topology restriction of Theorem~\ref{thm::PrivacyAlg2} using additive perturbation signals and preserve privacy for agents whose all communication signals (in-coming and out-going) are available to the malicious agent. Herein, we show that this mechanism has no contribution and the malicious agent can still derive local information asymptotically. A comprehensive protection via perturbation signals suggests adding a signal $g^i(k):\mathbb{Z}_{\geq0}\to\real$ to the \review{transmitted} message of every agent $i\in\VV$, and additive signal $f_1^i(k):\mathbb{Z}_{\geq0}\to\real$ and $f_2^i(k):\mathbb{Z}_{\geq0}\to\real$ to the right hand side of~\eqref{eq::Alg2a} and~\eqref{eq::Alg2b}, respectively. However, without loss of generality, the effect of all these signals can be captured via adding only a dynamic perturbation signal $f^i(k):\mathbb{Z}_{\geq0}\to\real$, $i\in\VV$ to the right hand side of~\eqref{eq::Alg2b} as follows,
    \begin{subequations}\label{eq::perturbedAlg2}
    \begin{align}
        &v^i(k+1)= v^i(k)+\Delta \sum\nolimits_{j=1}^{N}\vectsf{a}_{ij}(x^{i}(k)-x^{j}(k)),\label{eq::perturbedAlg2a}\\
        &x^i(k+1)=x^i(k)+\Delta(-(x^i(k)-\mathsf{r}^i)\nonumber\\
        &\quad\quad  -\sum\nolimits_{j=1}^{N}\vectsf{a}_{ij}(x^{i}(k)-x^{i}(k))-v^i(k)+f^i(k)).\label{eq::perturbedAlg2b}
    \end{align}
    \end{subequations}
    Instead of using a particular perturbation signal as in~\cite{manitara2013privacy, mo2016privacy,he2018privacy}, we follow the approach in~\cite{NR-SSK:19arxiv} and first investigate for what set of perturbation signals, which we call \emph{admissible perturbation signals}, the final convergence point of the algorithm is not perturbed. It is natural that any necessary condition defining the admissible perturbation signal is known to all the agents. Then, we show that by knowing a necessary condition on the perturbation, the malicious agent can employ an observer  to obtain the reference input regardless of what the exact additive admissible perturbation signal the agent uses.

\begin{theorem}[A necessary condition on admissible perturbation signals]\label{thm::Alg2perturbedDis}
Let the interaction topology $\GG$ of the agents implementing algorithm~\eqref{eq::perturbedAlg2} with $\Delta\in(0,\min\{2,\bar{\Delta}\})$, initialized at $x^i(0)\in\real$ and $v^i(0)=0$, $i\in\VV$, be strongly connected and weight-balanced digraph. Let the knowledge set of malicious agent $1$ be as in Definition~\ref{asm::knowledgesetadversary}. A necessary condition for preserving the average consensus, i.e., $\lim_{k\to\infty}x^i(k)=\mathsf{r}^{\text{avg}}$ for $i\in\VV$ is
    \begin{align}\label{eq::perturbedAlg2conditionsDisCol}
        &\lim_{k\to\infty}\sum\nolimits_{i=1}^{N}\sum\nolimits_{m=0}^{k}(1-\Delta)^{k-m}f^{i}(m)=0,\quad i\!\in\!\VV.
    \end{align}
\end{theorem}
\begin{proof}
Using the change of variable introduced in Section II,~\eqref{eq::perturbedAlg2} can be written as
\begin{subequations}
\begin{align*}
    \mathsf{q}_1(k+1)&= \mathsf{q}_1(k)=0,\\
    \vectsf{q}_{2:N}(k+1)&= (1-\Delta)\vectsf{q}_{2:N}(k) +\Delta\,\vectsf{\mathfrak{R}}^\top\vectsf{f}(k),\\
    \mathsf{p}_1(k+1)&= (1-\Delta)\mathsf{p}_1(k)+\Delta\,\vectsf{\mathfrak{r}}^\top\vectsf{f}(k),\\
    \vectsf{p}_{2:N}(k+1)&=(\vect{I}-\Delta\lL^{+}) \vectsf{p}_{2:N}(k)\!-\!\Delta\vectsf{q}_{2:N}(k)+\Delta\vectsf{\mathfrak{R}}^\top\vectsf{f}(k),
\end{align*}
\end{subequations}
where we used $\mathsf{q}_1(0)\!=\!\sum_{i=0}^{N}v^i(0)=0$. Then, 
\begin{align}
    \mathsf{p}_1(k)=(1-\Delta)^{k}\mathsf{p}_1(0)
    -\sum\nolimits_{m=0}^{k}(1-\Delta)^{k-m}\Delta(-\vectsf{\mathfrak{r}}^\top\vectsf{f}(m)).\label{eq::p1}
\end{align}
Since $\left[\begin{smallmatrix}\mathsf{p}_1 \\\vectsf{p}_{2:N}\end{smallmatrix}\right]=\left[\begin{smallmatrix}\vectsf{\mathfrak{r}}^\top \\ \vectsf{\mathfrak{R}}^\top\end{smallmatrix}\right](\vect{x}-\mathsf{r}^{\text{avg}}\vectsf{1}_N)$, the necessary condition for reaching average consensus then is $\lim_{k\to\infty}\vectsf{p}(k)=\vectsf{0}$. Given that $\lim_{k\to\infty}(1-\Delta)^k=0$, then it follows from~\eqref{eq::p1} that \eqref{eq::perturbedAlg2conditionsDisCol} is a necessary condition for $\lim_{k\to\infty}\mathsf{p}_1(k)=0$.
\end{proof}
\begin{remark}{\rm
 Condition~\eqref{eq::perturbedAlg2conditionsDisCol} that defines the admissible perturbation signals couples the choices of all the agents. In a decentralized setting, since there is no third-party to assign local perturbation signals $f^i$, for agents to satisfy~\eqref{eq::perturbedAlg2conditionsDisCol} without any coordination among themselves, agents choose their admissible perturbations according to  \begin{align}\label{eq::perturbedAlg2conditionsDis}
        &\lim_{k\to\infty}\sum\nolimits_{m=0}^{k}(1-\Delta)^{k-m}f^{i}(m)=0,\quad i\in\VV.
    \end{align}}\boxend
\end{remark}

In our next statement, we investigate whether a malicious agent can derive the local reference value of an agent if it knows condition~\eqref{eq::perturbedAlg2conditionsDis} and all the \review{transmitted} messages to and from the agent.

\begin{theorem}[\review{Use of locally chosen} perturbation signals in~\eqref{eq::perturbedAlg2} does not increase privacy protection level]\label{thm::adversarialobserverdesignperturbed}
     Let the interaction topology $\GG$ of the agents implementing algorithm~\eqref{eq::perturbedAlg2} with $\Delta\in(0,\min\{2,\bar{\Delta}\})$, initialized at $x^i(0)\in\real$ and $v^i(0)=0$, $i\in\VV$, be strongly connected and weight-balanced digraph. Suppose agents are implementing \review{locally chosen admissible perturbation signals that satisfy  \eqref{eq::perturbedAlg2conditionsDis}, and} let the knowledge set of malicious agent $1$ be as in Definition~\ref{asm::knowledgesetadversary}. If agent $1$ is the in-neighbor of agent $i\in\VV\backslash\{1\}$ and all its out-neighbors, then agent $1$ can employ the observer
     \begin{subequations}\label{eq::Alg2obs}
     \begin{align}
        \hat{v}^{i}(k+1)&=\hat{v}^{i}(k)-\Delta\sum\nolimits_{j=1}^{N}\vectsf{a}_{ij}(x^{i}(k)-x^j(k)),\\
        \hat{x}^{i}(k+1)&=\hat{x}^{i}(k)+\Delta\big(\sum\nolimits_{j=1}^{N}\vectsf{a}_{ij}(x^{i}(k)-x^j(k))\nonumber\\
        &+\hat{v}^{i}(k)-\hat{x}^{i}(k)\big),\\
        \hat{z}^{i}(k)&=\hat{x}^{i}(k)+x^{i}(k),\label{eq::hatZ}
    \end{align}
    \end{subequations}
    initialized at $\hat{v}^i(0)\!=\!\hat{x}^i(0)\!=\!0$ to have $\hat{z}^i(k)\!\to\!\mathsf{r}^i$ as $k\!\to\!\infty$.
\end{theorem}

\begin{proof}
  By substitution, from~\eqref{eq::hatZ} we obtain $$\hat{z}^{i}(k+1)=\hat{z}^{i}(k)+\Delta(-\hat{z}^{i}(k)-(v^{i}(k)-\hat{v}^{i}(k))+\mathsf{r}^i+f^{i}(k)),$$
    which gives us $$\hat{z}^{i}(k)=(1-\Delta)^{k}\hat{z}^{i}(0)+\Delta\sum_{m=0}^{k}(1-\Delta)^{k-m}(\mathsf{r}^i+f^{i}(m)).$$
    We note that using the sum of geometric series we can write $$\sum_{m=0}^{k}(1-\Delta)^{k-m}=\frac{1-(1-\Delta)^{k+1}}{\Delta}.$$ Then, since $|1-\Delta|<1$ and given the necessary condition~\eqref{eq::perturbedAlg2conditionsDis} the observer converges to the local reference value of agent $i$, i.e., $\hat{z}^i(k)\to\mathsf{r}^i$ as $k\to\infty$.
\end{proof}
We proved that an additive perturbation signal has no contribution in the level of privacy provided by algorithm~\eqref{eq::Alg2}. 

\section{Demonstration Examples}\label{sec::num}

To provide a context to appreciate the implications of our privacy preserving solution for average consensus problem versus the method of \cite{mo2016privacy}, we consider an optimal power dispatch (OPD) problem over the undirected connected graph in Fig.~\ref{fig::digraph}(b).~\cite{mo2016privacy} is a representative of the common method of privacy preservation via additive noises for algorithm~\eqref{eq::Alg1}. We conduct our study over an undirected connected graph since the results in~\cite{mo2016privacy} are established only for such graphs. In our OPD problem a group of generators interacting over the graph of Fig.~\ref{fig::digraph}(b) are expected to meet the demand of $\mathcal{P}_D=1500$ MW ($p^1+\cdots+p^5=1500$) while minimizing their collective cost $\sum_{i=1}^5 f^i(p^i)$. The parameters of the local cost  $f^i(p^i)=\frac{1}{2\beta^i}(p^i+\alpha^i)^2+\gamma^i$, $i\in\{1,\cdots,5\}$, are chosen from the IEEE bus 118 generator list according to corresponding components of $\{\alpha^i\}_{i=1}^4=\{188.3,592.5,2567.2,1793.3,2567.2\}$, $\{\beta^i\}_{i=1}^5=\{7.17,45.9,208.2,166.6,208.2
\}$. 
Here $\alpha^i$ and $\beta^i$ are supposed to be the private value of each agent $i\in\VV$. The optimal solution of this problem for each agent $i\in\{1,\cdots,5\}$ is given by $p^{i\star}=\beta^i\frac{\mathcal{P}_D+\sum_{i=1}^N\alpha^i}{\sum_i^N\beta^i}-\alpha^i$~\cite{AJW-BFW-GBS:14}. To generate this solution in a distributed manner, let us assume that these $N=5$ agents employ two static average consensus algorithms to obtain $\bar{\alpha}=\frac{1}{N}\sum_{i=1}^N\alpha^i$ and $\bar{\beta}=\frac{1}{N}\sum_i^N\beta^i$.
\begin{figure}[t!]
    \centering
    {\!\!\!\!\!\includegraphics[trim= 1pt 5pt 16pt 5pt ,clip,width=0.45\linewidth]{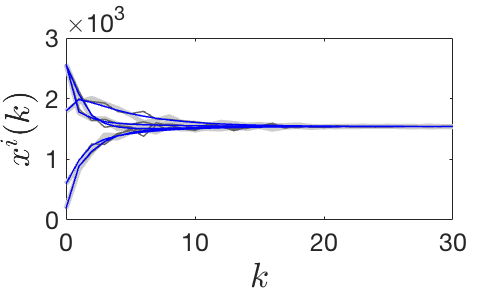}\label{Fig1a}}\quad
    {\!\!\!\!\!\includegraphics[trim= 1pt 5pt 16pt 5pt ,clip,width=0.45\linewidth]{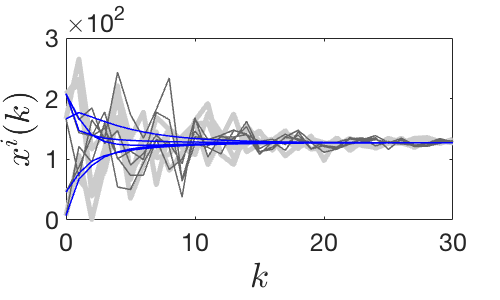}\label{Fig1b}}
    \caption{{\small Trajectories of the average consensus algorithms to obtain $\bar{\alpha}$ (left plot) and $\bar{\beta}$ (right plot). The thick gray lines and the thin gray lines show trajectories of two executions of method M1\review{~\cite[Corollary 1]{mo2016privacy}}, each employing a different noise realization. The blue lines show the trajectories of the agreement state of algorithm~\eqref{eq::Alg2}.
    }}
    \label{fig::Fig1}
\end{figure}
Then, by knowing $N$ and $\mathcal{P}_D$, agents have all the information to obtain $p^{i\star}$. Fig.~\ref{fig::digraph}(b) is 
used in the numerical example of \cite{mo2016privacy} where agent $5$ is the malicious agent, so as we assume here.
The privacy of agents $\{1,2,3\}$ is preserved if agents use algorithm~\eqref{eq::Alg1} augmented by the additive noise method of~\cite[Corollary 1]{mo2016privacy} (hereafter referred to as  method M1) or algorithm~\eqref{eq::Alg2} (see Theorem~\ref{thm::PrivacyAlg2}). Let us assume that when using the method M1, the agents use the same $\phi=0.9$ as~\cite[Section VI]{mo2016privacy} and given the value of their parameters, they agree to use an additive Gaussian noise with mean $0$ and standard deviation $\sigma=100$. Note that these values should be common for all agents, and thus agents need to coordinate to choose them. Fig.~\ref{fig::Fig1} shows the trajectories generated by method M1 for two different executions (each uses a different noise realization) overlaid over the trajectories generated by Algorithm~\eqref{eq::Alg2}. As seen, the trajectories of method M1 are quite noisy and convergence is slower. Using method M1, malicious agent $5$ cannot obtain the exact value of $\{\alpha^i\}_{i=1}^3$ and $\{\beta^i\}_{i=1}^3$ because as shown in~\cite[Fig.4]{mo2016privacy} the covariance of a maximum likelihood estimator that agent $5$ uses to obtain the private reference value of $\{1,2,3\}$ has a steady state non-vanishing value, see Fig.~\ref{fig::covariance_mo}. However, agent $5$ knows that in $99.7$\% of the times ($3\sigma$ rule) the error rate to obtain each $
\{\alpha^i\}_{i=1}^3$ and $\{\beta^i\}_{i=1}^3$ is respectively $(0.972\%,0.618\% ,0.071\%)$ and $(25.512\%,7.972\%,0.879\%)$ according to the computed normalized covariances $P_{ii}$ as in Fig.~\ref{fig::covariance_mo}. Given the numerical values of these parameters, this level of protection gives a good approximate value of the optimal operating point of the supposedly private agents $\{1,2,3\}$ to the malicious agent. 
On the other hand, the privacy preservation guarantees that algorithm~\eqref{eq::Alg2} provides for agents $\{1,2,3\}$ is stronger, as the malicious agent cannot obtain any estimate on the private values of the agent. See Fig.~\ref{fig::Alg1} for an example scenario, where two alternative cases of reference input signals for $\alpha^i$ denoted by 
$\{\alpha_2^i\}_{i=1}^5\!=\!\{-1311.6,3592.5,1067.2,1793.3,2567.2\}$, and $\{\alpha_3^i\}_{i=1}^5\!=\!\{1688.3,-2407.4,4067.2,1793.3,2567.2\}$.
where $\frac{1}{N}\sum_{i=1}^5\!{\alpha}_j^i\!=\!\frac{1}{N}\sum_{i=1}^5\!\alpha_1^i$, $j\!\in\!\{2,3\}$ result in exactly the same \review{transmit} message to the malicious agent $5$.   Therefore, agent $5$ cannot distinguish between these different references for agents $\{1,2,3\}$. Here, $\{\alpha_1^i\}_{i=1}^5$ are the actual inputs given in the OPD problem definition earlier.
\begin{figure}[t!]
    \centering
   {\includegraphics[width=0.55\linewidth]{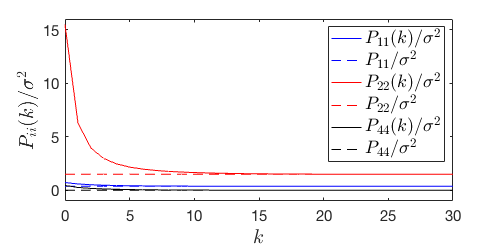}\label{Fig2a}}
   \caption{{\small Normalized covariances of the maximum likelihood estimator that the malicious agent uses to obtain an estimate of reference input of the agents implementing method M1. The dashed line is the steady-state value and the solid lines show the time history. This plot is the same as~\cite[Fig. 4]{mo2016privacy} where $\sigma\!=\!1$ is~used.}}
   \label{fig::covariance_mo}
\end{figure}

\begin{figure}[t!]
    \centering
    \subfloat[$\{x_j^i(k)\}_{i=1}^5,j\!\in\!\{1,2,3\}$ vs. $k$ are shown, respectively, by the dashed blue, the thin gray and the thick gray lines.]{\includegraphics[width=0.45\linewidth]{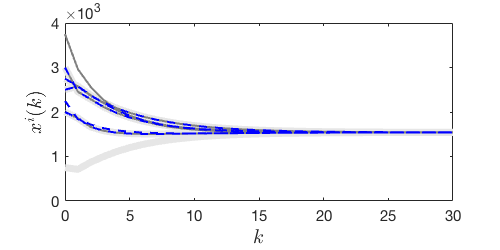}\label{fig::MoCovariance-a}}~~~
    \subfloat[$e_x^i(k)
    =x_1^i(k)-x_2^i(k)$
    vs. $k$]{\includegraphics[width=.26\linewidth]{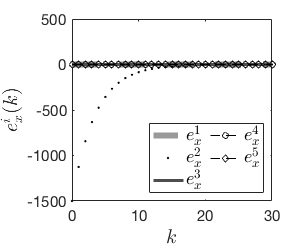}\label{fig::MoCovariance-b}}~~~
    \subfloat[$e_x^i(k)
    =x_1^i(k)-x_3^i(k)$
    vs. $k$]{\includegraphics[width=.26\linewidth]{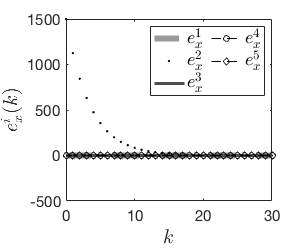}\label{fig::MoCovariance}}
    \caption{{\small Three executions of~\eqref{eq::Alg2} yield exactly the same received signals by the malicious agent $5$, while converging to the same average value of $\frac{1}{N}\sum_{i=1}^5\alpha^i$.}}
    \label{fig::Alg1}
\end{figure}
\begin{figure}
    \centering
\end{figure}

\section{Conclusion}
We considered the problem of privacy preservation in the static average consensus problem. This problem normally is solved by proposing privacy preservation mechanism that are added to the popular first order Laplacian-based algorithm. These mechanisms come with computational overheads or pre-coordinating among the agents to choose the parameters of the algorithm. They also alter the transient response of the algorithm. In this paper we showed that an alternative algorithm that is proposed in the literature in the context of dynamic average consensus problem can be a simple solution for privacy preservation for average consensus problem. The advantage of our proposed solution over existing results in the literature was to provide a stronger notion of privacy while rendering a similar transient and convergence behavior to that of the well-known Laplacian average consensus algorithm, having no need for coordination to choose the parameters of the algorithm, and no extra~computations. 

\bibliographystyle{ieeetr}

\end{document}